\documentclass[submission,copyright,creativecommons]{eptcs}
\providecommand{\event}{DCM 2018} 
\usepackage{underscore}           
\usepackage{color}
\usepackage{xcolor}
\usepackage{amssymb}
\usepackage{amsmath}
\usepackage{amsthm}
\usepackage{stmaryrd}
\usepackage{enumerate}
\usepackage{listings}
\usepackage{wrapfig,subfig}
\usepackage{float}
\floatstyle{boxed}
\restylefloat{figure}   
 
\newcommand{\NamedRuleOL}[3]{{\tiny{\textsc{(#1)}}}
\displaystyle                  
{%
\begin{array}{l}#2\end{array}}\           
\begin{array}{l}#3\end{array}     
}

\newtheorem{theorem}{Theorem}
\newtheorem{corollary}{Corollary}
\newtheorem{lemma}{Lemma}
\newcommand{\refToFigure}[1]{Fig.\ref{fig:#1}}
\newcommand{\refToSection}[1]{Sect.\ref{sect:#1}}
\newcommand{\refToLemma}[1]{Lemma~\ref{lemma:#1}}

\newcommand{\refToTheorem}[1]{Theorem~\ref{theo:#1}}

\newcommand{\Space}{\hskip 0.4em}
\newcommand{\BigSpace}{\hskip 1.5em}

\definecolor{light-gray}{gray}{0.80}

\newcommand{\store}[1]{\colorbox{light-gray}{\texttt{#1}}}

\newcommand{\Tuple}[1]    {\langle{#1}\rangle}
\newcommand{\Pair}[2]     {\Tuple{{#1},{#2}}}

\newcommand{\Reduct}[2]{#1_{|#2}}

\newcommand{\NamedRule}[4]{{\tiny{\textsc{(#1)}}}
\displaystyle                  
\frac{%
\begin{array}{l}#2\end{array}%
}{#3}        
{{\footnotesize \begin{array}{l}#4\end{array}}}    
}
\newcommand{\rn}[1]{{\scriptsize (\textsc{#1})}}					

\newenvironment{grammatica}{$\begin{array}{lcll}}{\end{array}$}
\newcommand{\produzione}[3]{#1&::=&#2&\mbox{#3}}

\newcommand{\seguitoproduzione}[2]{&&#1&\mbox{#2}}
\newcommand{\terminale}[1]{\texttt{#1}}
\newcommand{\metavariable}[1]{{\it #1}}

\newcommand{\x}{\metavariable{x}}
\newcommand{\y}{\metavariable{y}}
\newcommand{\z}{\metavariable{z}}
\newcommand{\xs}{\metavariable{xs}}
\newcommand{\m}{\metavariable{m}}
\newcommand{\e}{\metavariable{e}}
\newcommand{\etilde}{\hat{\e}}
\newcommand{\estilde}{\widehat{\es}}
\newcommand{\iotas}{\iota_1,\dots,\iota_n}
\newcommand{\ctxHat}{\widehat{\ctx}}
\newcommand{\es}{\metavariable{es}}
\newcommand{\C}{\metavariable{C}}
\newcommand{\D}{\metavariable{D}}
\newcommand{\f}{\metavariable{f}}
\newcommand{\dec}{\metavariable{d}}
\newcommand{\decs}{\metavariable{ds}}
\newcommand{\decstilde}{\widehat{\decs}}
\newcommand{\dectilde}{\widehat{\dec}}
\newcommand{\this}{\terminale{this}}
\newcommand{\Field}[2]{#1\ #2}
\newcommand{\FieldAccess}[2]{#1\terminale{.}#2}
\newcommand{\MethCall}[3]{#1{\terminale{.}}#2\terminale{(}#3\terminale{)}}
\newcommand{\FieldAssign}[3]{#1\terminale{.}#2\terminale{=}#3}
\newcommand{\ConstrCall}[2]{\terminale{new}\ #1\terminale{(}#2\terminale{)}}
\newcommand{\Block}[2]{\terminale{\{}#1\;#2\terminale{\}}}
\newcommand{\Dec}[3]{#1\,#2\terminale{=}#3;}

\newcommand{\Param}[2]{#1\ #2}
\newcommand{\Sequence}[2]{#1\terminale{;}#2}

\newcommand{\ExpMem}[2]{#1{\mid}#2}

\newcommand{\capsule}{\terminale{a}}
\newcommand{\Type}[2]{#2^{#1}}
\newcommand{\qualifier}{\metavariable{q}}

\newcommand{\T}{\metavariable{T}}
\newcommand{\X}{\metavariable{X}} 
\newcommand{\Y}{\metavariable{Y}} 
\newcommand{\BlockLab}[3]{ \terminale{\{}^{#3} #1\;#2\terminale{\}}}
\newcommand{\IsCapsule}[1]{\aux{caps}(#1)}
\newcommand{\dv}{\metavariable{dv}}
\newcommand{\dvs}{\metavariable{dvs}}
\newcommand{\ctx}{{\cal{E}}}
\newcommand{\ctxP}{{\cal E}'}
\newcommand{\Ctx}[1]{\ctx[#1]}
\newcommand{\CtxP}[1]{\ctxP[#1]}
\newcommand{\CtxHat}[1]{\ctxHat[#1]}
\newcommand{\valCtx}{\ctx_v}
\newcommand{\ValCtx}[1]{\valCtx[#1]}
\newcommand{\blockCtx}{\ctx_b}
\newcommand{\blockCtxHat}{\widehat{\ctx_b}}
\newcommand{\BlockCtx}[1]{\blockCtx[#1]}
\newcommand{\BlockCtxHat}[1]{\blockCtxHat[#1]}
\newcommand{\emptyctx}{[\ ]}
\newcommand{\val}{\metavariable{v}}

\newcommand{\vs}{\metavariable{vs}}

\newcommand{\creduce}[2]{#1\Longrightarrow#2}
\newcommand{\creducestar}[2]{#1\Longrightarrow^\star#2}

\newcommand{\reduce}[2]{#1\longrightarrow^{}#2}

\newcommand{\reduceOS}[2]{#1{\stackrel{\aux{c}}{\longrightarrow}}#2}
\newcommand{\reduceN}[3]{#1\longrightarrow^{#3}#2}
\newcommand{\updateCtx}[3]{\blockCtx^{#2.#3{=}#1}}
\newcommand{\UpdateMem}[4]{#1^{#2.#3=#4}}
\newcommand{\UpdateCtx}[4]{\updateCtx{#1}{#2}{#3}[#4]}
\newcommand{\reducestar}[2]{#1\longrightarrow^\star#2}
\newcommand{\Subst}[3]{#1[#2/#3]}
\newcommand{\SubstDot}[5]   {#1[#2/#3\dots#4/#5]}

\newcommand{\congruence}[2]{{#1}\cong{#2}}
\newcommand{\aux}[1]{\textsf{#1}}
\newcommand{\fields}[1]{\aux{fields}(#1)}
\newcommand{\method}[2]{{\aux{meth}(#1,#2)}}
\newcommand{\FV}[1]{\aux{FV}(#1)}
\newcommand{\HB}[1]{\aux{HB}(#1)}
\newcommand{\Inner}[1]{\aux{inner}(#1)}
\newcommand{\dom}[1]{\aux{dom}(#1)}
\newcommand{\img}[1]{\aux{img}(#1)}

\newcommand{\extractDec}[2]{\aux{get}(#1,#2)}

\newcommand{\correspondence}[3]{#1\stackrel{_{#2}}{\leftrightsquigarrow}#3}

\title{A Syntactic Model of Mutation and Aliasing}
\author{Paola Giannini\thanks{This original research has the financial support of the Universit\`a  del Piemonte Orientale.}
\institute{Computer Science Institute, DiSIT, Universit\`{a} del Piemonte Orientale\\
Italy}
\email{paola.giannini@uniupo.it}
\and
Marco Servetto
\institute{School of Engineering and Computer Science,Victoria University of Wellington\\
New Zealand}
\email{servetto@ecs.vuw.ac.nz}
\and Elena Zucca
\institute{DIBRIS, Universit\`a di Genova\\
Italy}
\email{elena.zucca@unige.it}
}
\def\titlerunning{A Syntactic Model of Mutation and Aliasing}
\def\authorrunning{P. Giannini, M. Servetto \& E. Zucca}
\begin{document}
\maketitle

\begin{abstract}
Traditionally, semantic models of imperative languages use an auxiliary structure which mimics memory. In this way, ownership and other encapsulation properties need to be reconstructed from the graph structure of such global memory. 
We present an alternative \emph{syntactic} model where memory is encoded as part of the program rather than as a separate resource. This means that execution can be modelled by just rewriting source code terms, as in semantic models for functional programs. 
Formally, this is achieved by the block construct, introducing local variable declarations, which play the role of memory when their initializing expressions have been evaluated. 
In this way, we obtain a language semantics which directly represents at the syntactic level constraints on aliasing, allowing simpler reasoning about related properties. 
To illustrate this advantage, we consider the issue, widely studied in the literature, of characterizing an isolated portion of memory, 
which cannot be reached through external references. In the syntactic model, closed block values, called \emph{capsules}, provide a simple representation of isolated portions of memory, and capsules can be safely moved to another location in the memory, without introducing sharing, by means of \emph{affine} variables.
We {prove that the syntactic model can be encoded in the conventional one, hence efficiently implemented.}
\end{abstract}

\section{Introduction}
In languages with state and  mutations,  keeping control of aliasing relations is a key issue for correctness. This is  {exacerbated} by  concurrency mechanisms, since side-effects in one thread can affect the behaviour of another thread, hence unpredicted aliasing can induce unplanned/unsafe communication.

For these reasons, the last few decades have seen considerable interest in type systems for controlling sharing and interference, to make programs easier to maintain and understand, notably using \emph{qualifiers} to restrict the usage of references \cite{ZibinEtAl10,GordonEtAl12,NadenEtAl12,ClebschEtAl15}.

In particular, in an ongoing stream of work \cite{ServettoZucca15,GianniniSZ16,GianniniSZ17,GianniniSZ17a,
GianniniSZ18,GianniniRSZ19a,GianniniSZC19b}, we have adopted an innovative execution model \cite{ServettoLindsay13,CapriccioliSZ15} for imperative languages which, 
differently from traditional  ones, {is a \emph{reduction relation on language terms}. That is, we do not add an auxiliary structure which mimics 
physical memory, but such structure is encoded in the language itself.} Whereas this makes no difference from a programmer's point of view, it is important on the foundational side, since, as will be informally illustrated below, language semantics directly represents at the syntactic level constraints on aliasing, allowing simpler reasoning about related properties. 

In this paper, we focus on the operational model itself, rather than on type systems, and formalize its relation with the conventional model, where an auxiliary global structure mimics memory. 

To informally introduce this syntactic calculus, we show examples of reduction sequences.
The main idea is to use variable declarations to directly represent the memory.
That is, a declared variable is not replaced by its value, as in
standard \texttt{let}, but the association is kept and used when necessary, as
it happens, with different aims and technical problems, in cyclic lambda
calculi \cite{AriolaFelleisen97,MaraistEtAl98,AriolaBlom02}.
Assuming a program  (class table) where class \lstinline{C} has two fields
\lstinline{f1} and \lstinline{f2} of type \lstinline{D}, and class
\lstinline{D} has a field \lstinline{f} of type \lstinline{D}, the term
\begin{lstlisting}[basicstyle=\ttfamily\footnotesize,backgroundcolor=\color{white}]
D x=new D(y); D y=new D(x); C w={D z=new D(z); x.f=x; new C(z,z)}; w.f1
\end{lstlisting} 
starts with two declarations that can be seen as a memory consisting of two mutually referring objects.
Then there is a declaration whose right-hand-side needs to be evaluated and the final expression returns 
the value of a field of the object associated with this last variable. The reduction of the term is as follows.

\begin{lstlisting}[basicstyle=\ttfamily\footnotesize,backgroundcolor=\color{white}]
$\store{D x=new D(y); D y=new D(x);}$ C w={D z=new D(z); x.f=x; new C(z,z)}; w.f1 $\longrightarrow$
$\store{D x=new D(x); D y=new D(x);}$ C w={D z=new D(z); new C(z,z)}; w.f1 $\longrightarrow$
$\store{D x=new D(x); D y=new D(x); D z=new D(z);}$ C w={new C(z,z)}; w.f1 $\cong$
$\store{D x=new D(x); D y=new D(x); D z=new D(z); C w=new C(z,z);}$ w.f1 $\longrightarrow$
$\store{D x=new D(x); D y=new D(x); D z=new D(z); C w=new C(z,z);}$ z $\longrightarrow$
$\store{D z=new D(z);}$ z
\end{lstlisting} 

Evaluation proceeds left to right. {We emphasize at each step the declarations which can be seen
as the store (in grey).}
We start evaluating the right-hand-side of the declaration for \lstinline{w},
by updating the field \lstinline{f} of \lstinline{x}. Then, in order to evaluate  the field access \lstinline{w.f1}, {we need to move the declaration of \lstinline{z} outside of the inner block. A block with no declarations is considered equivalent to its body, as expressed by $\cong$.}
Now the field access \lstinline{w.f1} can be performed, getting \lstinline{z}, 
and the last step removes declarations (memory) which are not reachable from \lstinline{z}, giving as final result a memory consisting of only one cyclic object.

To illustrate how aliasing constraints are directly represented in this syntactic model, we consider an important example: how to characterize an isolated portion of memory, 
which cannot be reached through external references. {This allows programmers (and static analyses)
to identify portions of memory that can be safely handled by a thread.} This property has been widely studied in the literature, under different names and variants, such as \emph{isolated} \cite{GordonEtAl12},
\emph{unique} \cite{Boyland10} and \emph{externally
unique}~\cite{ClarkeWrigstad03}, \emph{balloon}
\cite{Almeida97}, \emph{island} \cite{DietlEtAl07}.  

In the syntactic calculus, block values with no free variables, called \emph{capsules}, provide a simple representation of portions of memory which are trivially isolated.
For instance, in the following example:
\begin{lstlisting}[basicstyle=\ttfamily\footnotesize,backgroundcolor=\color{white}]
$\rn{ex1}\BigSpace\store{D x=new D(x); D y=new D(x);}$ C$^\capsule$ w={D z=new D(z); new C(z,z)}; x.f=x
\end{lstlisting} 
the right-hand side of the declaration of \lstinline{w} is a capsule. To allow the programmer to \emph{safely rely} on the fact that \lstinline{w} denotes an isolated portion of memory, which nobody else in the program can affect, we introduce the qualifier $\capsule$, as shown above, for \emph{affine} variables/parameters, which should be initialized with capsules.

If the term is, instead:
\begin{lstlisting}[basicstyle=\ttfamily\footnotesize,backgroundcolor=\color{white}]
$\rn{ex2}\BigSpace\store{D x=new D(x); D y=new D(x);}$ C$^\capsule$ w={D z=new D(z); new C(z,y)}; x.f=x
\end{lstlisting} 
then the inner block is \emph{not} a capsule, hence its use to inizialize an affine variable is an error to be prevented, otherwise a programmer using \Q@w@ would erroneously rely on the capsule property. In the syntactic calculus, this error can be easily detected at runtime in a modular way, that is, by looking only at the block itself. For instance, in \rn{ex1} the runtime check succeeds, whereas in  \rn{ex2} it fails, hence normal execution cannot proceed, since the block contains the free variable \lstinline{y}.\footnote{In the formalization presented in this paper, the term is simply stuck. A different choice could be to introduce explicit \emph{error} terms. In our previous work \cite{ServettoZucca15,GianniniSZ16,GianniniSZ17,GianniniSZ17a,
GianniniSZ18} we have designed type systems which are able to statically prevent such error.} 

Note that, in the conventional calculus, detecting that a reference  denotes an isolated portion of memory requires, instead, a dependency analysis involving surrounding code and memory.  
For instance, in the examples \rn{ex1} and \rn{ex2} above, execution will reach respectively the following step:
\begin{quote}
$\rn{ex1}\BigSpace\ExpMem{\texttt{C$^\capsule$ w = new C($\iota_\z$,$\iota_\z$); $\iota_\x$.f=$\iota_\x$}}{\mu}$\\
$\rn{ex2}\BigSpace\ExpMem{\texttt{C$^\capsule$ w = new C($\iota_\z$,$\iota_\y$); $\iota_\x$.f=$\iota_\x$}}{\mu}$
\end{quote}
where the domain of the memory $\mu$ are \emph{object identifiers} $\iota$, modeling \emph{global} names which, differently from variables, do not support shadowing and $\alpha$-renaming.
Here $\mu=\iota_\x\mapsto\ConstrCall{\D}{\iota_\x}, \iota_\y\mapsto\ConstrCall{\D}{\iota_\x},\iota_\z\mapsto\ConstrCall{\D}{\iota_\z}$. It is clear that the two situations cannot be distinguished by only looking locally at the initialization expression of \lstinline{w}. Instead, we must check that the memory portion reachable from such initialization expression is isolated, that is, cannot be reached from references used in other parts of the program.
This is true for \rn{ex1}, since there is no sharing between $\iota_\z$ and the reference $\iota_\x$ used in the external code \texttt{$\iota_\x$.f=$\iota_\x$}. 
For \rn{ex2}, instead, there is sharing between $\iota_\y$ and $\iota_\x$, and indeed the execution of such external code affects the  portion of  memory reachable from \texttt{new C($\iota_\y$,$\iota_\z$)}.
In the general case, detecting sharing through dependency analysis is an expensive check, which could even be impossible in a distributed environment.

As said above, by using an affine variable/parameter $\x$ the programmer can safely rely on the fact that $\x$ denotes an isolated portion of memory. On the other hand, the capsule property should be preserved when $\x$ is used. This is ensured by the constraint that affine variables/parameters can be used at most once, and by a special reduction semantics, motivated and described below.

As already shown, for (non-affine) variable declarations reduction is as follows:
\begin{lstlisting}[basicstyle=\ttfamily\footnotesize,backgroundcolor=\color{white}]
$\store{D x=new D(x); D y=new D(x);}$ C w={D z=new D(z); new C(z,z)}; w.f1 $\longrightarrow$
$\store{D x=new D(x); D y=new D(x); D z=new D(z);}$ C w={new C(z,z)}; w.f1 $\cong$
$\store{D x=new D(x); D y=new D(x); D z=new D(z); C w=new C(z,z);}$ w.f1 $\longrightarrow$ $\ldots$
\end{lstlisting} 
That is, the block is flattened, hence the capsule property is lost. This corresponds to the fact that, if we get access to a portion of memory through an ordinary variable, then sharing could be introduced through such variable, hence, in particular, a portion of memory which was isolated is not guaranteed to remain such. 
To preserve the property, affine variables have a special semantics, which allow a capsule to be moved to another location in the memory, or passed as argument to a method.
In the above example, reduction would be as follows:
\begin{lstlisting}[basicstyle=\ttfamily\footnotesize,backgroundcolor=\color{white}]
$\store{D x=new D(x); D y=new D(x);}$ C$^\capsule$ w={D z=new D(z); new C(z,z)}; w.f1 $\longrightarrow$
$\store{D x=new D(x); D y=new D(x);}$ {D z=new D(z); new C(z,z)}.f1 $\longrightarrow$
$\store{D x=new D(x); D y=new D(x);}$ {D z=new D(z); new C(z,z).f1} $\longrightarrow$
$\store{D x=new D(x); D y=new D(x);}$ {D z=new D(z); z} $\longrightarrow$
$\store{D z=new D(z);}$ z
\end{lstlisting} 
Differently from the previous reduction, the block occurring as right-hand side of the declaration of \lstinline{w} is \emph{not} flattened, but, rather, used to replace \lstinline{w}. That is, affine variables have a substitution semantics. Then, the field access is propagated inside the block to be performed.

In \refToSection{conventional} and \refToSection{syntactic} we define the conventional and the syntactic calculus, respectively. 
In \refToSection{preservation} we show that the syntactic model can be encoded in the conventional one, hence efficiently implemented, and {prove} that the dynamic semantics is preserved by the encoding.
Finally, in \refToSection{conclu} we draw some conclusions.

\section{The conventional calculus}\label{sect:conventional}
We illustrate our approach in the context of calculi with an object-oriented flavour, inspired {by}
Featherweight Java \cite{IgarashiEtAl01} (FJ for short). This is only a presentation choice:
the ideas and results of the paper could be rephrased in
different imperative calculi, e.g., supporting data type constructors and
reference types. For the same reason, we omit features such as inheritance and
late binding, which are orthogonal to our focus.

The conventional calculus is given in \refToFigure{conventional}.  It is similar to other imperative variants of FJ which can be found in the literature \cite{AhernY05,BiermanP03,BettiniDS10}.
 We assume sets
of \emph{variables} $\x, \y, \z$, \emph{class names} $\C, \D$, \emph{field
names} $\f$, and \emph{method names} $\m$.  We adopt the convention that a
metavariable which ends in \metavariable{s} is implicitly defined as a
(possibly empty) sequence in which elements may or may not be separated by commas. 
In particular, $\decs$ (and $\dvs$ in the next section) are sequences of $\dec$ (and $dv$)
and $\es$, $\vs$ and $\xs$ are sequences of $\e$, $\val$ and $\x$ separated by commas.

\begin{figure}[h]
{\small
\begin{grammatica}
\produzione{\e}{\x\mid\FieldAccess{\e}{\f}\mid\FieldAssign{\e}{\f}{\e'}\mid\ConstrCall{\C}{\es}\mid\MethCall{\e}{\m}{\es}\mid\Block{\decs}{\e}\mid\iota}{expression}\\
\produzione{\dec}{\Dec{\C}{\x}{\e}}{declaration}\\ 
\\
\produzione{\val}{\iota}{value}\\
\produzione{\ctx}{\emptyctx\mid\FieldAccess{\ctx}{\f}\mid\FieldAssign{\ctx}{\f}{\e'}\mid\FieldAssign{\iota}{\f}{\ctx}\mid\ConstrCall{\C}{\vs,\ctx,\es}}{evaluation context}\\
\seguitoproduzione{\mid\MethCall{\ctx}{\m}{\es}\mid\MethCall{\iota}{\m}{\vs,\ctx,\es}\mid\Block{\Dec{\C}{\x}{\ctx}\ \decs}{\e}}{}
\end{grammatica}
{
\bigskip
\hrule
\bigskip
\begin{math}
\begin{array}{l}
{\NamedRuleOL{alpha}{\congruence{\Block{\decs\ \Dec{\C}{\x}{\e}\ \decs'}{\e'}}{\Block{\Subst{\decs}{\y}{\x}\ \Dec{\C}{\y}{\Subst{\e}{\y}{\x}}\ \Subst{\decs'}{\y}{\x}}{\Subst{\e'}{\y}{\x}}}}{}}
\BigSpace\NamedRuleOL{block-elim}{\congruence{\Block{}{\e}}{\e}}{}
\end{array}
\end{math}}
\bigskip
\hrule
\bigskip
\begin{math}
\begin{array}{l}
\NamedRule{ctx}{\creduce{\ExpMem{\e}{\mu}}{\ExpMem{\e'}{\mu'}}}{\creduce{\ExpMem{\Ctx{\e}}{\mu}}{\ExpMem{\Ctx{\e'}}{\mu'}}}{}\BigSpace
\NamedRuleOL{field-access}{\creduce{\ExpMem{\FieldAccess{\iota}{\f_i}}{\mu}}{\ExpMem{\val_i}{\mu}}}
{\mu(\iota)=\ConstrCall{\C}{\val_1,\dots,\val_n}\\
\fields{\C}=\Field{\C_1}{\f_1}\dots\Field{\C_n}{\f_n}{\ \wedge\ 1\leq i\leq n}}\\[4ex]
\NamedRuleOL{field-assign}{\creduce{\ExpMem{\FieldAssign{\iota}{\f_i}{\val}}{\mu}}{\ExpMem{\val}{\UpdateMem{\mu}{\iota}{i}{\val}}}}{\mu(\iota)=\ConstrCall{\C}{\vs}\\
\fields{\C}=\Field{\C_1}{\f_1}\dots\Field{\C_n}{\f_n}{\ \wedge\ 1\leq i\leq n}}\\[4ex]
\NamedRuleOL{new}{\creduce{\ExpMem{\ConstrCall{\C}{\vs}}{\mu}}{\ExpMem{\iota}{\Subst{\mu}{\ConstrCall{\C}{\vs}}{\iota}}}}{\iota\not\in\dom{\mu}}\\[4ex]
\NamedRuleOL{invk}{\creduce{\ExpMem{\MethCall{\iota}{\m}{\val_1,\dots,\val_n}}{\mu}}{\ExpMem{\Block{\Dec{\C_1}{\this}{\iota}\ \Dec{\C_1}{\x_1}{\val_1}\dots\Dec{\C_n}{\x_n}{\val_n}}{\e}}}{\mu}}
{\mu(\iota)=\ConstrCall{\C}{\vs}\\
\method{\C}{\m}=\Pair{\x_1\dots\x_n}{\e}}\\[4ex]
\NamedRuleOL{dec}{\creduce{\ExpMem{\Block{\Dec{\C}{\x}{\iota}\ \decs}{\e}}{\mu}}{\ExpMem{\Subst{\Block{\decs}{\e}}{\iota}{\x}}{\mu}}}{}
\end{array}
\end{math}
}
\caption{Conventional calculus}\label{fig:conventional}
\end{figure}

An expression can be a variable (including the special variable $\this$
denoting the receiver in a method body), a field access, a field assignment, a
constructor invocation, a method invocation, or a block consisting of a sequence of local variable declarations and a
body. In addition,  a (runtime) expression can be an \emph{object identifier} $\iota$. 
Blocks are included to have a more direct correspondence with the syntactic calculus.
In a block, a declaration specifies a type (class name), a
variable and an initialization expression. We assume that in well-formed blocks
there are no multiple declarations for the same variable{, that is, $\decs$ can
  be seen as a map from variables to expressions: $\dom{\decs}$ denotes the set of variables declared in $\decs$
  and $\decs(\x)$ the initialization expression for $\x$ in $\decs$, if any. 
  
In the examples,
we generally omit the brackets of the outermost block, and abbreviate
$\Block{\Dec{\T}{\x}{\e}}{\e'}$ by $\Sequence{\e}{\e'}$ when $\x$ does not
occur free in $\e'$. We also assume integer constants, which are not included in the formalization. 

  {Expressions are identified modulo congruence, denoted by $\congruence{}{}$, defined as the
smallest congruence satisfying the axioms in the mid section of \refToFigure{conventional}.  Rule \rn{alpha} is the usual $\alpha$-conversion. The condition
$\x,\y\not\in\dom{\decs\,\decs'}$ is implicit by well-formedness of blocks. Rule \rn{block-elim} identifies  a block without declarations with its body.}
  
The class table is abstractly modelled by the following functions:
\begin{itemize}
  \item $\fields{\C}$ gives, for each declared class $\C$, the sequence 
    $\Field{\C_1}{\f_1}\dots\Field{\C_n}{\f_n}$ of its fields 
    declarations.
      \item $\method{\C}{\m}$ gives, for each method $\m$ declared in class $\C$, the pair 
    consisting of its parameters and body.
\end{itemize} 

The reduction relation $\creduce{}{}$ is defined on pairs $\ExpMem{\e}{\mu}$ where a \emph{memory} $\mu$ is a finite map from object identifiers $\iota$ into object states of shape $\ConstrCall{\C}{\vs}$.
Values are object identifiers (we do not identify the two sets since, extending the language, values would be extended to include, e.g., primitive values {such as integers}). 

Evaluation contexts and reduction rules are straightforward. {In rule \rn{field-assign}, we} denote by $\UpdateMem{\mu}{\iota}{i}{\val}$ the memory where the $i$-th field of the object state associated to $\iota$ has been replaced by $\val$.
{In rule \rn{invk}, we take advantage of the block construct to provide a modular semantics, where
a method call is
reduced to a block where declarations of the appropriate type for $\this$ and the
parameters are initialized with the receiver and the arguments, respectively, and the body is the method body. Indeed, this rule plus a sequence of applications of rule \rn{Dec} is equivalent to the standard FJ rule}
${\creduce{\ExpMem{\MethCall{\iota}{\m}{\val_1,\dots,\val_n}}{\mu}}{\ExpMem{\SubstDot{\Subst{\e}{\iota}{\this}}{\val_1}{\x_1}{\val_n}{\x_n}}}{\mu}}$.}
Local variable declarations have the standard substitution semantics, and are elaborated in the usual left-to-right order (no recursion is allowed).

 
 \section{The syntactic calculus}\label{sect:syntactic}

The syntax of the expressions, given in \refToFigure{pure-syntax}, is the same as the conventional calculus, except that runtime expressions (object identifiers) are not needed. 
To lighten the notation, we use the same metavariables.

\begin{figure}[h]
{\small
\begin{grammatica}
\produzione{\e}{\x\mid\FieldAccess{\e}{\f}\mid\FieldAssign{\e}{\f}{\e'}\mid\ConstrCall{\C}{\es}\mid\MethCall{\e}{\m}{\es}\mid\BlockLab{\decs}{\e}{\X}}{expression}\\
\produzione{\dec}{\Dec{\T}{\x}{\e}}{declaration}\\ 
\produzione{\T}{\Type{\qualifier}{\C}}{declaration type}\\
\produzione{\qualifier}{\epsilon\mid\capsule}{optional qualifier}\\
\\
\produzione{\val}{\x\mid\BlockLab{\dvs}{\x}{\X}}{value}\\
\produzione{\dv}{\Dec{\C}{\x}{\ConstrCall{\C}{\xs}}}{evaluated declaration}\\
\produzione{\ctx}{\emptyctx\mid
\FieldAccess{\ctx}{\f}\mid\FieldAssign{\ctx}{\f}{\e'}\mid\FieldAssign{\x}{\f}{\ctx}\mid
\ConstrCall{\C}{{\xs},\ctx,\es}\mid\MethCall{\ctx}{\m}{\es}\mid
\MethCall{{\x}}{\m}{\vs,\ctx,\es}
\mid\blockCtx{}}{evaluation context}\\
\produzione{\blockCtx}{\BlockLab{\dvs\ \Dec{\C}{\y}{\ctx}\ \decs}{\e}{\X}\mid \BlockLab{\dvs}{\ctx}{\X}
}{block context}\\
\produzione{\valCtx}{\FieldAccess{\emptyctx}{\f}\mid\FieldAssign{\emptyctx}{\f}{\e'}\mid\FieldAssign{\x}{\f}{\emptyctx}\mid\ConstrCall{\C}{\xs,\emptyctx,\es}\mid{\MethCall{\emptyctx}{\m}{\es}}}{value context}
\end{grammatica}}
\caption{Syntactic calculus: syntax, values, and evaluation contexts}
\label{fig:pure-syntax}
\end{figure}

Moreover, some annotations are inserted in terms. Namely:
\begin{itemize}
\item Local variable declarations (and method parameters) are possibly annotated with a qualifier $\capsule$,
{which, if present, indicates that} the variable is \emph{affine}.   An affine variable can occur at most once in its scope, and should be
initialized with a \emph{capsule}, that is, an isolated portion of store. In
this way, it can be used as a temporary reference, to ``move'' a capsule to
another location in the store, without introducing sharing. 
\item Blocks are annotated with a set $\X$ of variables, assumed to be a subset of the declared variables. During reduction, if a block $\BlockLab{\decs}{\e}{\X}$ should reduce to a capsule, only declarations of variables which are not in $\X$ can be safely moved outside of the block, see rule \rn{move-dec}. In this paper, since our focus is on the operational model, we do not care about how block annotations are generated. Of course, a trivial overapproximation consists in taking as $\X$ the set of all declared variables; a better approximation is taking only those which are  used (that is, have some free occurrence in initialization expressions/body), or, {even better}, {only those which} are transitively used by the body, in the sense formally defined below. We have shown in previous work \cite{GianniniSZ17,GianniniSZ17a,GianniniSZ18} that through a type and effect system  it is possible to {obtain a much more precise approximation, computing $\X$ to be the set of the} variables which \emph{will be possibly connected with the final result of the block}. 
\end{itemize}

A sequence $\dvs$ of \emph{evaluated declarations} plays the role of the memory
in the conventional calculus, that is, each $\dv$ can be seen
as an association of an \emph{object state} $\ConstrCall{\C}{\xs}$ to a reference. 

A value is a reference to an object, possibly enclosed in a block where all declarations are evaluated (hence, correspond to a local memory).\\
We assume that, in a block value $\BlockLab{\dvs}{\x}{\X}$, $\dvs\neq\epsilon$ and $\Reduct{\dvs}{\x}=\dvs$, where, given a sequence of declarations $\decs\equiv\Dec{\T_1}{\x_1}{\e_1}\dots\Dec{\T_n}{\x_n}{\e_n}$ and an expression $\e$,  
$\Reduct{\decs}{\e}$ are the declarations of variables
(transitively) used by $\e$, that is, free either in $\e$ or in some $\e_i$ such that $\x_i$ is transitively used by $\e$.

We write $\FV{\decs}$ and $\FV{\e}$ for the free variables of a sequence of
declarations and of an expression, respectively, and $\Subst{\X}{\y}{\x}$,
$\Subst{\decs}{\y}{\x}$, and $\Subst{\e}{\y}{\x}$ for the capture-avoiding
variable substitution on a set of variables, a sequence of declarations, and an
expression, respectively, all defined in the standard way. 

In the syntactic calculus, capsules can be characterized in a very simple way: indeed, a value is a capsule, written $\IsCapsule{\val}$, if it is  a closed block value, that is, of shape $\Block{\dvs}{\x}$ with no free variables. The above requirement that all local variables must be transitively used by $\x$ is needed, indeed, since otherwise a block value containing unused free variables, e.g., $\Block{\Dec{\C}{\x}{\ConstrCall{\C}{}}\ \Dec{\D}{\y}{\ConstrCall{\D}{\z}} }{\x}$ would be not recognized to be a capsule. Unused evaluated declarations are removed by rule \rn{garbage}. 

Evaluation contexts $\ctx$ are mostly standard. Note that values are assumed to be references, apart from arguments of method calls, which are allowed to be block values. This models the fact that block values (hence, capsules) are first-class values which can be passed to methods. However, they need to be ``opened'' when we perform an actual operation on them. We distinguish, among evaluation context, \emph{block contexts}, $\blockCtx$, having the shape of a block, which play a special role in the reduction rules.

The {\em hole binders} of a context $\ctx$, dubbed $\HB{\ctx}$, are the variables bound by the context, defined by:
\begin{itemize}
\item $\HB{\emptyctx}=\emptyset$, $\HB{\FieldAccess{\ctx}{\f}}=\cdots=\HB{\MethCall{{\x}}{\m}{\vs,\ctx,\es}}=\HB{\ctx}$
\item $\HB{\BlockLab{\dvs\ \Dec{\C}{\y}{\ctx}\ \decs}{\e}{\X}}=\dom{\dvs\, \decs}\cup\{\y\}\cup\HB{\ctx}$ and $\HB{\BlockLab{\dvs}{\ctx}{\X}}=\dom{\dvs}\cup\HB{\ctx}$.
\end{itemize}
Moreover, given $\blockCtx=\BlockLab{\dvs\ \Dec{\C}{\y}{\ctx}\ \decs}{\e}{\X}$ or $\blockCtx=\BlockLab{\dvs}{\ctx}{\X}$, we define 
\begin{itemize}
\item $\extractDec{\blockCtx}{\x}=\dvs(\x)$, i.e., the object state associated to $\x$ in $\dvs$, if $\x\in\dom{\dvs}$, and  
\item $\Inner{\blockCtx}=\HB{\ctx}$, i.e.,  the set of variables declared in the direct subcontext $\ctx$ of $\blockCtx$.
\end{itemize}

A \emph{value context} $\valCtx$ is an evaluation context with the shape of either a field access, or a field assignment, or a constructor invocation, or a method invocation, where the hole (expected to be filled with a block value) is a direct subterm (the receiver in the last case).

Expressions are identified modulo congruence, denoted by $\cong$, defined as the
smallest congruence satisfying the axioms in \refToFigure{pure-cong}.  {Rules \rn{alpha} and \rn{block-elim} are as in the conventional calculus.}
Rule \rn{reorder} states that we can move evaluated declarations in an
arbitrary order.  Note that, instead, $\decs$ and $\decs'$ cannot be swapped,
because this could change the order of side effects. 

\begin{figure}[h]
{\small 
\begin{math}
\begin{array}{l}
{\NamedRuleOL{alpha}{{\congruence{\BlockLab{\decs\ \Dec{\C}{\x}{\e}\ \decs'}{\e'}{\X}}{\BlockLab{\Subst{\decs}{\y}{\x}\ \Dec{\C}{\y}{\Subst{\e}{\y}{\x}}\ \Subst{\decs'}{\y}{\x}}{\Subst{\e'}{\y}{\x}}{{\Subst{\X}{\y}{\x}}}}}}{}}
\BigSpace{\NamedRuleOL{block-elim}{\congruence{\BlockLab{}{\e}{\emptyset}}{\e}}{}}\\[3ex]
\NamedRuleOL{reorder}{\congruence{
\BlockLab{
\decs\ \dv\ \decs'}{\e}{\X}}{
\BlockLab{
\dv\ \decs\ \decs' }{\e}{\X}}}{}
\end{array}
\end{math}}
\caption{Syntactic calculus: congruence rules}\label{fig:pure-cong}
\end{figure}

Reduction rules are given in \refToFigure{pure-red}. {We explicitly distinguish the relation $\reduceOS{}{}$ defined by \emph{computational rules} only, and reduction $\reduce{}{}$, defined as its contextual closure. In other words, in rule \rn{ctx} we assume that $\ctx$ is a maximal context. This simplifies proofs in \refToSection{preservation}. More importantly, in this way we can prevent reduction of a constructor call in a declaration context, see comments to rule \rn{new} below.}

\begin{figure}[h]
{\small 
\begin{math}
\begin{array}{l}
\NamedRule{ctx}{\reduceOS{\e}{\e'}}{\reduce{\Ctx{\e}}{\Ctx{\e'}}}{}\BigSpace
{\NamedRuleOL{new}
{  \reduce{ \Ctx{\ConstrCall{\C}{\xs}} }
   {  \Ctx{\BlockLab{\Dec{\C}{\x}{\ConstrCall{\C}{\xs}}}{\x}{{\{\x\}}}} }    }
{\ctx\ {\neq\ \CtxP{\BlockLab{\dvs\ \Dec{\C}{\y}{\emptyctx}\ \decs}{\e}{\X}}}}}
\\[4ex]
\NamedRuleOL{field-access}{\reduceOS{\BlockCtx{\FieldAccess{\x}{\f_i}}}{\BlockCtx{\x_i}}}{
{\extractDec{\blockCtx}{\x}=\ConstrCall{\C}{\x_1,\dots,\x_n} \wedge\ {\x\not\in\Inner{\blockCtx}}}\\
\fields{\C}=\Field{\C_1}{\f_1}\dots\Field{\C_n}{\f_n}{\ \wedge\ 1\leq i\leq n}\\
\x_i\not\in\Inner{\blockCtx}}
\\[4ex]
{\NamedRuleOL{field-assign}{
\reduceOS{\BlockCtx{\FieldAssign{\x}{\f_i}{\y}}}{{\UpdateCtx{\y}{\x}{i}{\y}}}
}{
{\extractDec{\blockCtx}{\x}=\ConstrCall{\C}{{\xs}}\wedge\ {\x\not\in\Inner{\blockCtx}}}\\
\fields{\C}=\Field{\C_1}{\f_1}\dots\Field{\C_n}{\f_n}{\ \wedge\ 1\leq i\leq n}\\
\y\not\in\Inner{\blockCtx}}}
\\[4ex]
\NamedRuleOL{invk}{\reduceOS{\BlockCtx{\MethCall{{\x}}{\m}{\val_1,\ldots,\val_n}}}{\BlockCtx{\Block{\Dec{\C}{\this}{{\x}}\, \Dec{\Type{\qualifier_1}{\C_1}}{\x_1}{\val_1}, \ldots,\Dec{\Type{\qualifier_n}{\C_n}}{\x_n}{\val_n}}{\e}}}}{
\!\!\!\!\!\extractDec{\blockCtx}{\x}{=}\ConstrCall{\C}{\xs}\\
\!\!\!\!\! \x\not\in\Inner{\blockCtx}\\
\!\!\!\!\!\method{\C}{\m}{=}{\Pair{\Param{\Type{\qualifier_1}{\C_1}}{\x_1},\ldots,\Param{\Type{\qualifier_n}{\C_n}}{\x_n}}{\e}}}
\\[4ex]
{\NamedRuleOL{alias-elim}{\reduceOS{\BlockLab{\dvs\ \Dec{\C}{\x}{\y}\ \decs }{\e}{\X}}{\BlockLab{\dvs\ \Subst{\decs}{\y}{\x}}{\Subst{\e}{\y}{\x}}{\X\setminus\{\x\}}}}{}}
\\[4ex]
{\NamedRuleOL{affine-elim}{\reduceOS{\BlockLab{\dvs\ \Dec{\Type{\capsule}{\C}}{\x}{\val}\ \decs }{\e}{X}}{{\BlockLab{\dvs\ \Subst{\decs}{\val}{\x}}{\Subst{\e}{\val}{\x}}{\X\setminus\{\x\}}}}}{\IsCapsule{\val}}}
\\[4ex]
\NamedRuleOL{garbage}
{
\reduceOS{
\BlockLab{\dvs\ \decs}{\e}{\X}}{\BlockLab{\decs}{\e}{{\X}\setminus{\dom{\dvs}}}}}
{
\begin{array}{l}
(\FV{\decs}\cup\FV{\e})\cap\dom{\dvs}=\emptyset
\end{array}}
\\[4ex]
\NamedRuleOL{move-dec}
{\reduceOS{\BlockLab{{\dvs}\ \Dec{{\Type{\qualifier}{\C}}}{\x}{\BlockLab{{\dvs'}\ {\decs}}{\e}{\X}}\ \decs'}{\e'}{{\Y}}}
{\BlockLab{{\dvs}\ {\dvs'}\ \Dec{{\Type{\qualifier}{\C}}}{\x}{\BlockLab{{\decs}}{\e}{\X}}\ \decs'}{\e'}{{\Y}}}}
{\FV{{\dvs'}}\cap\dom{{\decs}}=\emptyset\\
\FV{{\dvs}\ \decs'\ {\e'}}{\cap}\dom{{\dvs'}}{=}\emptyset\\ 
{\qualifier=\capsule\Rightarrow\dom{\dvs'}{\cap}\X{=}\emptyset}}
\\[4ex]
\NamedRuleOL{move-body}
{\reduceOS{\BlockLab{{\dvs}}{\BlockLab{{\dvs'}\ {\decs}}{\e}{\X}}{\Y}}{\BlockLab{\dvs\ \dvs'}{\BlockLab{\decs_2}{\e}{\X}}{{\Y}}}}
{\FV{{\dvs'}}\cap\dom{{\decs}}=\emptyset\\
\FV{{\dvs}}\cap\dom{{\dvs'}}=\emptyset}
\\[4ex]
\NamedRuleOL{move-subterm}{\reduceOS
{\ValCtx{\BlockLab{{\dvs}\ {\dvs'}}{{\x}}{\X}}}
{\BlockLab{{\dvs}}{\ValCtx{\BlockLab{ {\dvs'}}{{\x}}{\X\setminus{\dom{{\dvs}}}}}}{\X\cap\dom{{\dvs}}}}}
{\FV{{\dvs}}\cap\dom{ {\dvs'}}=\emptyset\\
\FV{\valCtx}\cap\dom{{\dvs}}=\emptyset}
\end{array}
\end{math}}
\caption{Syntactic calculus: reduction rules}\label{fig:pure-red}
\end{figure}
Rule \rn{ctx} is the usual contextual closure.

In rule \rn{new}, a constructor invocation where all arguments are references 
is reduced to an elementary block where a new object is allocated. However, this reduction is not allowed when the constructor invocation occurs as right-hand side of a declaration, to prevent non-terminating reduction sequences such as the following:\\
\centerline{$\BlockLab{\Dec{\C}{\x}{\ConstrCall{\C}{\x}}}{\x}{}\longrightarrow\BlockLab{\Dec{\C}{\x}{\BlockLab{\Dec{\C}{\x}{\ConstrCall{\C}{\x}}}{\x}{}}}{\x}{}\longrightarrow\cdots$}
 {Note that, if the constructor invocation occurs, instead, on the right-side of a declaration of an affine variable, the
 rule is applicable and does not produce a non-terminating reduction sequence. For instance:\\
\centerline{
 $\BlockLab{\Dec{\Type{\capsule}{\C}}{\x}{\ConstrCall{\C}{\x}}}{\x}{}\longrightarrow\BlockLab{\Dec{\Type{\capsule}{\C}}{\x}{\BlockLab{\Dec{\C}{\x}{\ConstrCall{\C}{\x}}}{\x}{}}}{\x}{}$}
 and rule \rn{new} is no longer applicable. }

In rule \rn{field-access}, a field access of shape
$\FieldAccess{\x}{\f}$ is evaluated in the block context containing the first enclosing (evaluated) declaration for $\x$, as expressed by the first
side condition and the definition of $\extractDec{\blockCtx}{\x}$. 
 The fields of the class $\C$ of
$\x$ are retrieved from the class table.  If $\f$ is the name of a
field of $\C$, say, the $i$-th, then the field access is reduced to the
reference $\x_i$ stored in this field.  
The condition $\x_i\not\in\Inner{\blockCtx}$
ensures that there are no inner declarations for $\x_i$ (otherwise
$\x_i$ would be erroneously bound).  This can always be obtained by rule
\rn{alpha} of \refToFigure{pure-cong}.
For instance, assuming a class table where class \lstinline{A} has an
\lstinline{int} field, and class \lstinline{B} has an \lstinline{A} field \lstinline{f},
without this side condition, the term:
\begin{lstlisting}
A a=new A(0); B b=new B(a); {A a=new A(1); b.f}
\end{lstlisting}
{would reduce to}
\begin{lstlisting}
A a=new A(0); B b=new B(a); {A a=new A(1); a}
\end{lstlisting}
whereas this reduction is forbidden by the side condition. However, by rule \rn{alpha} of \refToFigure{pure-cong}, the term is congruent to one that can
be reduced to
\begin{lstlisting}
A a=new A(0); B b=new B(a); {A a1=new A(1); a}
\end{lstlisting}

In rule \rn{field-assign}, a field assignment of shape $\FieldAssign{\x}{\f}{\y}$ is evaluated in the block context containing the first enclosing (evaluated) declaration for $\x$, as expressed by the first side condition. 
The fields of the class $\C$ of
$\x$ are retrieved from the class table. 
If $\f$ is the name of
a field of $\C$, say, the $i$-th, then this first enclosing declaration is
updated, by replacing the $i$-th constructor argument by $\y$ obtaining the declaration
$\Dec{\C}{\x}{\ConstrCall{\C}{\x_1,\x_{i-1},\y,\x_{i+1},\dots,\x_n}}$ as
expressed by the notation $\updateCtx{\y}{\x}{i}$ (whose obvious formal
definition is omitted).  Analogously to rule \rn{field-access}, we have the side
condition that  $\y\not\in\Inner{\blockCtx}$. This side condition, requiring that there
are no inner declarations for $\y$, prevents scope extrusion,
since if $\y\in\Inner{\blockCtx}$, $\updateCtx{\y}{\x}{i}$ would take $\y$
outside the scope of its definition. 
For example, without this side condition, the term
\begin{lstlisting}
A a=new A(0); B b=new B(a); {A a1=new A(1); b.f=a1}
\end{lstlisting}
would reduce to 
\begin{lstlisting}
A a=new A(0); B b=new B(a1); {A a1=new A(1); a1}
\end{lstlisting}
which is not correct since \verb!a1! is a free variable. The rules \rn{move-dec} and
\rn{move-body}  (see below) can be used to move the
declaration of $\y$ outside its declaration block. So the term reduces, instead, to
\begin{lstlisting}
A a=new A(0); B b=new B(a); A a1=new A(1); b.f=a1
\end{lstlisting}
by applying first rule \rn{move-body}, and then {congruence rule} \rn{block-elim}. Now
the term correctly reduces to
\begin{lstlisting}
A a=new A(0); B b=new B(a1); A a1=new A(1); a1
\end{lstlisting}

In rule \rn{invk}, a method call of shape $\MethCall{\x}{\m}{\val_1,..,\val_n}$ is evaluated in the block context containing the first enclosing (evaluated) declaration for $\x$, as expressed by the first side condition. 
Method $\m$ of $\C$, if any, is retrieved from the class table.
The call is
reduced to a block where declarations of the appropriate type for $\this$ and the
parameters are initialized with the receiver and the arguments, respectively, and the body is the method body. 
If a parameter is affine, then the corresponding argument will be checked to be a capsule when the
formal parameter will be substituted with the associated value by rule \rn{affine-elim}, see below.  We assume that parameters which are affine occur at most once in the body of the method.

The following two rules eliminate declarations from a block.\\
In rule \rn{alias-elim}, a reference (non affine variable) $\x$ which is
initialized as an alias of another reference $\y$ is eliminated by replacing
all its occurrences. In rule \rn{affine-elim}, an affine variable is eliminated
by replacing its unique occurrence with the value associated to its
declaration. The rule can only be applied if such value is a capsule. Such side condition formally models
that execution includes a runtime check in this case, as it happens, e.g., in Java downcasts. Note also that, even ignoring the $\capsule$ qualifier in the latter, the two rules do not overlap, since a reference $\y$ is trivially not a capsule.

Rule \rn{garbage} states that we can remove an unused
sequence of evaluated declarations from a block. Note that it is only possible
to safely remove declarations which are evaluated, since they do not have
side effects.

With the remaining rules we can move a sequence of 
evaluated declarations from a block to
the directly enclosing block, as it happens with rules for
\emph{scope extension} in the $\pi$-calculus \cite{Milner99}.
 
In rules \rn{move-dec} and \rn{move-body}, the inner block is the right-hand
side of a declaration, or the body, respectively, of the enclosing block. The first two
side conditions ensure that moving the declarations {$\dvs'$} {does  cause
neither scope extrusion nor capture of free variables}. More precisely: the first
prevents moving outside a declaration {$\dvs'$} which depends on local
variables of the inner block. The second prevents capturing with {$\dvs'$} free
variables of the enclosing block.  Note that the second condition can be
obtained by $\alpha$-conversion of the inner block, but the first cannot.
Finally, when the block initializes an affine variable,
the third side condition of rule \rn{dec} forbids to move outside the
block declarations of variables that are in the annotation $\X$ of the block. Indeed, annotations can be rough or very precise, as discussed before, but in any case such that variables not in $\X$ cannot be possibly connected to the final result of the block, hence can be safely moved outside without ``breaking'' the capsule property (that the block should ultimately reduce to a closed expression).
In case of a non affine
declaration, instead, this is not a problem. 

Rule \rn{move-subterm} handles the cases when the inner block is a  subterm of a
field access, field assignment, constructor invocation, or method invocation.
Note that in this case the inner block is necessarily a (block) value. We use value contexts to
express all such cases in a compact way.

\section{Preservation of semantics}\label{sect:preservation}
In this section, for clarity, we use $\etilde$ and $\decstilde$ to range over expressions and sequences of declarations of the conventional calculus, which could include object identifiers.

We show that, if an expression has a reduction sequence in the syntactic calculus, then it has an ``equivalent'' reduction sequence in the conventional calculus. 
Note that the converse does not hold, since an expression which is stuck in the syntactic calculus, since a capsule runtime check (side condition of rule \rn{affine-elim}) fails, could reduce in the conventional calculus.
That is, in the syntactic calculus, by declaring  $\x$ affine, the programmer can rely on the fact that during computation $\x$ will always denote an isolated portion of memory\footnote{{Hence, nobody else in the program can affect memory reachable from $\x$, hence such memory can be safely handled by a thread.} }, otherwise an exception would be raised.  

In order to state the preservation results (\refToTheorem{step} and Corollary \ref{preservation}), we define a \emph{matching} relation $\correspondence{}{\rho}$ between expressions $\e$ of the syntactic calculus and pairs $\ExpMem{\etilde}{\mu}$ of the conventional calculus.  {The relation is labelled by an injective mapping $\rho$ from object identifiers in {the domain of} $\mu$ to variables.  Intuitively, $\correspondence{\e}{\rho}{\ExpMem{\etilde}{\mu}}$ holds if
 variables in the image of $\rho$ are all those declared in evaluated declarations in $\e$ and all such declarations are removed in $\etilde$. More precisely, since the same variable could be declared in different blocks inside an expression, the mapping $\rho$ should be from object identifiers to \emph{binding occurrences} of variables. To make the formal treatment simpler, we can assume that each variable is declared at most once in expressions of the syntactic calculus. {This can} be obtained by $\alpha$-renaming.

Formally the relation is inductively defined by the rules in \refToFigure{correspondence}. 

A variable in the syntactic calculus is matched in the conventional calculus by {a pair where the expression is} the corresponding  object identifier in $\rho$, if any, rule \rn{oid}, otherwise the variable itself, rule \rn{var}. The former case happens when the variable declaration is evaluated in the syntactic calculus. 

Rules \rn{field-access}, \rn{field-assign}, \rn{new}, and \rn{invk} just propagate matching to subterms. For $\es=\e_1, \ldots, \e_n$, {$\estilde=\etilde_1, \ldots, \etilde_n$,} we use $\correspondence{\es}{\rho}{\ExpMem{\estilde}{\mu}}$ to abbreviate $\correspondence{\e_i}{\rho}{\ExpMem{\etilde_i}{\mu}}\Space \forall i\in 1..n$.

In rule \rn{block}, a block $\BlockLab{\dvs\ \decs}{\e}{\x}$  in the syntactic calculus is matched by a block and memory $\ExpMem{\Block{\decstilde}{\etilde}}{\mu}$ of the conventional calculus if evaluated declarations $\dvs$ are matched by memory $\mu$, as expressed by the auxiliary judgment $\correspondence{\dvs}{\rho}{\mu}$, and  matching is propagated to other subterms. 
For $\decs=\dec_1\ldots\dec_n$, $\decstilde=\dectilde_1\ldots\dectilde_n$, we use $\correspondence{\decs}{\rho}{\ExpMem{\decstilde}{\mu}}$ to abbreviate $\correspondence{\dec_i}{\rho}{\ExpMem{\dectilde_i}{\mu}}\Space \forall i\in 1..n$, and analogously for $\correspondence{\dvs}{\rho}{\mu}$.

\begin{figure}[h]
{\small 
\begin{math}
\begin{array}{l}
\NamedRule{oid}{}{\correspondence{\x}{\rho}{\ExpMem{\iota}{\mu}}}{\rho(\iota) = \x}{\BigSpace\NamedRule{var}{}{\correspondence{\x}{\rho}{\ExpMem{\x}{\mu}}}{\x\not\in\img{\rho}}}\\[4ex]
\NamedRule{field-access}{\correspondence{\e}{\rho}{\ExpMem{\etilde}{\mu}}}{\correspondence{\FieldAccess{\e}{\f}}{\rho}{\ExpMem{\FieldAccess{\etilde}{\f}}{\mu}}}{}\BigSpace
\NamedRule{field-assign}{\correspondence{\e}{\rho}{\ExpMem{\etilde}{\mu}}\BigSpace\correspondence{\e'}{\rho}{\ExpMem{\etilde'}{\mu}}}{\correspondence{\FieldAssign{\e}{\f}{\e'}}{\rho}{\ExpMem{\FieldAssign{\etilde}{\f}{\etilde'}}{\mu}}}{}\\[4ex]
\NamedRule{new}{ 
\correspondence{\es}{\rho}{\ExpMem{\estilde}{\mu}}
}{\correspondence{\ConstrCall{\C}{\es}}{\rho}{\ExpMem{\ConstrCall{\C}{\estilde}}{\mu}}}{}\BigSpace
    \NamedRule{invk}{ \correspondence{\e}{\rho}{\ExpMem{\etilde}{\mu}\BigSpace\correspondence{\es}{\rho}{\ExpMem{\estilde}{\mu}}}
 }{\correspondence{\MethCall{\e}{\m}{\es}}{\rho}{\ExpMem{\MethCall{\etilde}{\m}{\estilde}}{\mu}}}{} \\[4ex]
    {\NamedRule{block}{\correspondence{\dvs}{\rho}{\mu}\BigSpace\correspondence{\decs}{\rho}{\ExpMem{\decstilde}{\mu}}\BigSpace\correspondence{\e}{\rho}{\ExpMem{\etilde}{\mu}}}{\correspondence{\BlockLab{\dvs\ \decs}{\e}{\X}}{\rho}{\ExpMem{\Block{\decstilde}{\etilde}}}{\mu}}{}
    \BigSpace\NamedRule{dec}{\correspondence{\e}{\rho}{\ExpMem{\etilde}{\mu}}}{\correspondence{\Dec{\Type{\qualifier}{\C}}{\x}{\e}}{\rho}{\ExpMem{\Dec{\C}{\x}{\etilde}}{\mu}}}{\Dec{\Type{\qualifier}{\C}}{\x}{\e}\neq\dv}}\\[5ex]
\NamedRule{ev-dec}{}{\correspondence{\Dec{\C}{\x}{\ConstrCall{\C}{\x_1, \ldots, \x_n}}}{\rho}{\mu}}{\mu(\rho^{-1}(\x))=\ConstrCall{\C}{\rho^{-1}(\x_1), \ldots, \rho^{-1}(\x_n)}}
\end{array}
\end{math}}
\caption{Matching relation between terms}\label{fig:correspondence}
\end{figure}

The matching relation can be extended, {in the obvious way,} to evaluation contexts of the syntactic calculus and pairs evaluation context and
memory of the conventional calculus, i.e., $\correspondence{\ctx}{\rho}{\ExpMem{\ctxHat}{\mu}}$. The formal definition is given in \refToFigure{ctx-correspondence}. 

\begin{figure}[h]
{\small 
\begin{math}
\begin{array}{l}
\NamedRule{C-empty}{}{\correspondence{\emptyctx}{\rho}{\ExpMem{\emptyctx}{\mu}}}{}\BigSpace
\NamedRule{C-field-access}{\correspondence{{\ctx}}{\rho}{\ExpMem{\ctxHat}{\mu}}}{\correspondence{\FieldAccess{\ctx}{\f}}{\rho}{\ExpMem{\FieldAccess{\ctxHat}{\f}}{\mu}}}{}\\[4ex]
\NamedRule{C-field-assign-l}{\correspondence{{\ctx}}{\rho}{\ExpMem{\ctxHat}{\mu}}\BigSpace{\correspondence{\e}{\rho}{\ExpMem{\etilde}{\mu} }}}{\correspondence{\FieldAssign{\ctx}{\f}{\e}}{\rho}{\ExpMem{\FieldAssign{\ctxHat}{\f}{\etilde}}{\mu}}}{}
\BigSpace
\NamedRule{C-field-assign-r}{\correspondence{{\ctx}}{\rho}{\ExpMem{\ctxHat}{\mu}}\BigSpace{\correspondence{\x}{\rho}{\ExpMem{\iota}{\mu} }}}{\correspondence{\FieldAssign{\x}{\f}{\ctx}}{\rho}{\ExpMem{\FieldAssign{\iota}{\f}{\ctxHat}}{\mu}}}{}\\[4ex]
\NamedRule{C-new}{\correspondence{{\ctx}}{\rho}{\ExpMem{\ctxHat}{\mu}}\BigSpace\correspondence{\es}{\rho}{\ExpMem{\estilde}{\mu}}\BigSpace\correspondence{\xs}{\rho}{\ExpMem{\iotas}{\mu}}}
{\correspondence{\ConstrCall{\C}{\xs, \ctx,\es}}{\rho}{\ExpMem{\ConstrCall{\C}{\iotas, \ctxHat, \estilde}}{\mu}}}{}\\[4ex]
    \NamedRule{C-invk-rcv}{\correspondence{{\ctx}}{\rho}{\ExpMem{\ctxHat}{\mu}}\BigSpace\correspondence{\es}{\rho}{\ExpMem{\estilde}{\mu}}}{\correspondence{\MethCall{\ctx}{\m}{\es}}{\rho}{\ExpMem{\MethCall{\ctxHat}{\m}{\estilde}}{\mu}}}{} 
\BigSpace
 \NamedRule{C-invk-arg}
  {\correspondence{{\ctx}}{\rho}{\ExpMem{\ctxHat}{\mu}}\BigSpace\correspondence{\es}{\rho}{\ExpMem{\estilde}{\mu}}\BigSpace\correspondence{\vs}{\rho}{\ExpMem{\iotas}{\mu}}}
  {\correspondence{\MethCall{\x}{\m}{\vs, \ctx, \es}}{\rho}{\ExpMem{\MethCall{\iota}{\m}{\iotas, \ctxHat, \estilde}}{\mu}}}{} \\[4ex]     
    \NamedRule{C-blk-dec}{\correspondence{\dvs}{\rho}{\mu}\BigSpace\correspondence{{\ctx}}{\rho}{\ExpMem{\ctxHat}{\mu}}\BigSpace\correspondence{\decs}{\rho}{\ExpMem{\decstilde}{\mu}}\BigSpace\correspondence{\e}{\rho}{\ExpMem{\etilde}{\mu}}}{\correspondence{\BlockLab{\dvs\ \Dec{\C}{\y}{\ctx}\ \decs}{\e}{\X}}{\rho}{\ExpMem{\Block{\Dec{\C}{\y}{\ctxHat}\ \decstilde}{\etilde}}{\mu}}}{}
  \BigSpace
\NamedRule{C-blk-body}{\correspondence{\dvs}{\rho}{\mu}\BigSpace\correspondence{{\ctx}}{\rho}{\ExpMem{\ctxHat}{\mu}}}{\correspondence{\BlockLab{\dvs}{\ctx}{\X}}{\rho}{\ExpMem{\ctxHat}{\mu}}}{}
\end{array}
\end{math}}
\caption{Matching relation between evaluation contexts}\label{fig:ctx-correspondence}
\end{figure}

It is easy to show the following lemma.
\begin{lemma}\label{lemma:evCtx}
$\correspondence{\Ctx{\e}}{\rho}{\ExpMem{\etilde'}{\mu}}$ if and only if
$\etilde'=\CtxHat{\etilde}$ such that $\correspondence{\ctx}{\rho}{\ExpMem{\ctxHat}{\mu}}$ and $\correspondence{{\e}}{\rho}{\ExpMem{\etilde}{\mu}}$.
\end{lemma}
{Note that, when $\Ctx{\e}$ is closed, all the free variables in $\e$ are declared in $\ctx$ and their declaration is evaluated.
Therefore, the previous lemma implies that} they are in the image of $\rho$, and so in $\etilde$ they match object identifiers.

The following lemma asserts the conditions on which a value in the syntactic calculus {is matched by a pairs of value and memory} in the
conventional one. 
\begin{lemma}\label{lemma:value}
Let $\val$ be such that  $\correspondence{\val}{\rho}{\ExpMem{\etilde}{\mu}}$ for some $\etilde$, $\rho$ and $\mu$. Then:
\begin{enumerate}
\item if $\val=\x$, then {either $\etilde=\iota$ for some $\iota$ and  $\rho(\iota)=\x$, or $\etilde=\x$}
\item if $\val=\BlockLab{\dvs}{\x}{\X}$, then $\etilde=\iota$ {and $\rho(\iota)=\x$}. 
 \end{enumerate}
\end{lemma}
\begin{proof}
\begin{enumerate}
\item {The judgment $\correspondence{\x}{\rho}{\ExpMem{\etilde}{\mu}}$ has been necessarily derived by rule either \rn{oid} or \rn{var}} in \refToFigure{correspondence}.
\item First observe that $\dvs$ cannot be empty and that $\Dec{\C}{\x}{\ConstrCall{\C}{\xs}}$ must be in $\dvs$. {The judgment}
$\correspondence{\BlockLab{\dvs}{\x}{\X}}{\rho}{\ExpMem{\etilde}{\mu}}$ {has been necessarily derived by} rule \rn{Block} in \refToFigure{correspondence}.
Therefore {$\correspondence{\dvs}{\rho}{\mu}$, and, since for each declaration in $\dvs$ we have applied rule $\rn{Ev-Dec}$,} all the
variables in $\dvs$ must be in the image of $\rho$ and $\etilde=\iota$ such that $\rho(\iota)=\x$.
\end{enumerate}
\end{proof}
{From the previous lemma we have that, if all the free variables of $\val$ are in the image of $\rho$, then $\val$ can only match an
object identifier.}

A step in the reduction of a {closed} expression $\e$ in the syntactic calculus can be simulated by a possibly empty sequence of reduction steps of 
the matching configuration $\ExpMem{\etilde}{\mu}$ in the conventional calculus. 
\begin{theorem}\label{theo:step}
{If $\FV{\e}=\emptyset$, $\correspondence{\e}{\rho}{\ExpMem{\etilde}{\mu}}$,} and $\reduce{\e}{\e'}$, then
$\creducestar{\ExpMem{\etilde}{\mu}}{\ExpMem{\etilde'}{\mu'}}$ such that $\correspondence{\e'}{\rho'}{\ExpMem{\etilde'}{\mu'}}$, for some $\etilde'$, $\rho'$, $\mu'$ such that $\rho\subseteq\rho'$, $\dom{\mu}\subseteq\dom{\mu'}$. 
\end{theorem}
\begin{proof}
If $\reduce{\e}{\e'}$, then either rule \rn{Ctx} or rule \rn{new} of \refToFigure{pure-red} were applied.\\
Consider first \underline{rule \rn{Ctx}}. Then $\e=\Ctx{\e_1}$, $\e'=\Ctx{\e'_1}$ and  $\reduceOS{\e_1}{\e_1'}$.
From $\correspondence{\e}{\rho}{\ExpMem{\etilde}{\mu}}$ and \refToLemma{evCtx}
we have that $\etilde=\CtxHat{\etilde_1}$ where $\correspondence{\ctx}{\rho}{\ExpMem{\ctxHat}{\mu}}$ and $\correspondence{{\e_1}}{\rho}{\ExpMem{\etilde_1}{\mu}}$.\\
By cases on the reduction rule of \refToFigure{pure-red} applied to reduce ${\e_1}$ to ${\e_1'}$. We consider
\rn{Field-Assign}, \rn{Affine-elim}, \rn{Move-Dec} and \rn{Move-subterm}. The proof for the other rules is similar.\\
Rule $\rn{Field-Assign}$. Then
\begin{enumerate}[(i)]
\item $\e_1=\BlockCtx{\FieldAssign{\x}{\f_i}{\y}}$, 
\item $\e'_1={\UpdateCtx{\y}{\x}{i}{\y}}$, 
\item $\extractDec{\blockCtx}{\x}=\ConstrCall{\C}{\x_1,\dots,\x_n}$, and
\item $\fields{\C}=\Field{\C_1}{\f_1}\dots\Field{\C_n}{\f_n}{\ \wedge\ 1\leq i\leq n}$.
\end{enumerate}
From $\correspondence{{\e_1}}{\rho}{\ExpMem{\etilde_1}{\mu}}$ and
\refToLemma{evCtx}
we have that $\etilde_1=\BlockCtxHat{\etilde_2}$ with  $\correspondence{\blockCtx}{\rho}{\ExpMem{\blockCtxHat}{\mu}}$  and $\correspondence{\FieldAssign{\x}{\f_i}{\y}}{\rho}{\ExpMem{\etilde_2}{\mu}}$ for some $\etilde_2$. 
Since $\e$ is a closed term we get that $\x$ and $\y$ are defined in the evaluated declarations of $\Ctx{\blockCtx}$. Therefore
$\etilde_2$ is $\FieldAssign{\iota}{\f_i}{\iota'}$ for some $\iota$ and $\iota'$  such that $\rho(\iota)=\x$ and  $\rho(\iota')=\y$. 
Moreover, from rules \rn{block} and \rn{ev-dec} and (iii){, it}
must be $\mu(\iota)=\ConstrCall{\C}{\iota_1,\dots,\iota_n}$ for some $\iota_i$ such that $\rho(\iota_i)=\x_i$ for $i\in 1,\dots,n$.
From (iv) {and} rule \rn{ctx} with premise \rn{invk} of \refToFigure{conventional} we get
$\creduce{\ExpMem{\BlockCtxHat{\FieldAssign{\iota}{\f_i}{\iota'}}}{\mu}}{\ExpMem{\BlockCtxHat{\iota'}}{{\UpdateMem{\mu}{\iota}{i}{\iota'}}}}$.
Since we have assumed no shadowing of variables, in $\ctx$ there is no declaration of $\x$ and so 
\begin{enumerate}[(a)]
\item $\correspondence{\ctx}{\rho}{\ExpMem{\ctxHat}{\UpdateMem{\mu}{\iota}{i}{\iota'}}}$.
\end{enumerate}
From  $\correspondence{\blockCtx}{\rho}{\ExpMem{\blockCtxHat}{\mu}}$, $\rho(\iota)=\x$, $\rho(\iota')=\y$ and  $\mu(\iota)=\ConstrCall{\C}{\iota_1,\dots,\iota_n}$ we get
\begin{enumerate}[(a)]\addtocounter{enumi}{1}
\item $\correspondence{\updateCtx{\y}{\x}{i}}{\rho}{\ExpMem{\blockCtxHat}{\UpdateMem{\mu}{\iota}{i}{\iota'}}}$.
\end{enumerate}
Finally from (a), (b), $\rho(\iota')=\y$ and
\refToLemma{evCtx} we derive that $\correspondence{\Ctx{{\UpdateCtx{\y}{\x}{i}{\y}}}}{\rho}{\ExpMem{\CtxHat{\BlockCtxHat{\iota'}}}{{\UpdateMem{\mu}{\iota}{i}{\iota'}}}}$.
\\
Rule $\rn{Affine-Elim}$. Then
\begin{enumerate}[(i)]
\item $\e_1=\BlockLab{\dvs\ \Dec{\Type{\capsule}{\C}}{\x}{\val}\ \decs }{\e_b}{X}$,
\item $\e'_1={\BlockLab{\dvs\ \Subst{\decs}{\val}{\x}}{\Subst{\e_b}{\val}{\x}}{\X\setminus\{\x\}}}$ and
\item $\IsCapsule{\val}$.
\end{enumerate}
From (iii) $\val$ is closed and so $\val=\BlockLab{\dvs'}{\y}{{\Y}}$ for some ${\Y}$, $\dvs'$, and $\y$. 
From $\correspondence{{\e_1}}{\rho}{\ExpMem{\etilde_1}{\mu}}$, (i), rule \rn{block} of \refToFigure{correspondence},
\refToLemma{value}.2 we have that 
\begin{enumerate}[(a)]
\item $\etilde_1=\Block{\Dec{\Type{}{\C}}{\x}{\iota}\ \decstilde }{\etilde_b}$, $\correspondence{\dvs\,\dvs'}{\rho}{\mu}$, $\correspondence{\decs}{\rho}{\ExpMem{\decstilde}{\mu}}$,
$\correspondence{\e_b}{\rho}{\ExpMem{\etilde_b}{\mu}}$,
\item $\rho(\iota)=\y$  and $\correspondence{\val}{\rho}{\ExpMem{\iota}{\mu}}$.
\end{enumerate}
Applying reduction rule \rn{Dec} of \refToFigure{conventional} we get $\creduce{\ExpMem{\Block{\Dec{\Type{}{\C}}{\x}{\iota}\ \decstilde }{\etilde_b}}{\mu}}{\ExpMem{\Subst{\Block{\decstilde}{\etilde_b}}{\iota}{\x}}{\mu}}$. Since $\Dec{\Type{\capsule}{\C}}{\x}{\val}$ is not an
evaluated declaration $\x\not\in\img{\rho}$. Therefore the occurrence of $\x$ in $\e_1$ is in the
matching relation with an occurrence of $\x$ in $\etilde_1$ and there is only one such occurrence. 
Therefore from (a) and (b) we have 
\begin{enumerate}[(a)]\addtocounter{enumi}{2}
\item $\correspondence{{\BlockLab{\dvs\ \Subst{\decs}{\val}{\x}}{\Subst{\e_b}{\val}{\x}}{\X\setminus\{\x\}}}}{\rho}{\ExpMem{\Subst{\Block{\decstilde}{\etilde_b}}{\iota}{\x}}{\mu}}$. 
\end{enumerate}
From $\correspondence{\ctx}{\rho}{\ExpMem{\ctxHat}{\mu}}$, (i), (a), (c) and \refToLemma{evCtx} we derive that 
$\correspondence{\Ctx{\e_1}}{\rho}{\ExpMem{\Ctx{\etilde_1}}{\mu}}$.
\\
Rule $\rn{Move-Dec}$. 
Then
\begin{enumerate}[(i)]
\item $\e_1={\BlockLab{{\dvs}\ \Dec{{\Type{\qualifier}{\C}}}{\x}{\BlockLab{{\dvs'}\ {\decs}}{\e_b}{\X}}\ \decs'}{\e'}{{\Y}}}$  and
\item $\e'_1={\BlockLab{{\dvs}\ {\dvs'}\ \Dec{{\Type{\qualifier}{\C}}}{\x}{\BlockLab{{\decs}}{\e_b}{\X}}\ \decs'}{\e'}{{\Y}}}$ 
\end{enumerate}
We consider two cases: 
\begin{enumerate}
\item  either $\congruence{{\BlockLab{\decs}{\e_b}{\X}}}{\ConstrCall{\C}{\x_1,\dots,\x_n}}$ for some $x_1,\dots,\x_n$ 
\item {or this is not the case.}
\end{enumerate}
Case 1. Then 
\begin{enumerate}[(i)]\addtocounter{enumi}{2}
\item $\e_1={\BlockLab{{\dvs}\ \Dec{{\Type{\qualifier}{\C}}}{\x}{\BlockLab{{\dvs'}}{\ConstrCall{\C}{\x_1,\dots,\x_n}}{\X}}\ \decs'}{\e'}{{\Y}}}$ and 
\item $\e'_1={\BlockLab{{\dvs}\ {\dvs'}\ \Dec{{\Type{\qualifier}{\C}}}{\x}{\ConstrCall{\C}{{\x_1,\dots,\x_n}}}\ \decs'}{\e'}{{\Y}}}$. 
\end{enumerate}
From $\correspondence{\e_1}{\rho}{\ExpMem{\etilde_1}{\mu}}$, \refToLemma{evCtx} and rules \rn{new} and \rn{block} of \refToFigure{correspondence}
we have that 
\begin{enumerate}[(a)]
\item $\etilde_1={\Block{\Dec{{\Type{}{\C}}}{\x}{{\ConstrCall{\C}{\iota_1,\dots,\iota_n}}}\ \decstilde'}{\etilde'}}$,
\item $\correspondence{\dvs\ \dvs'}{\rho}{\mu}$, $\correspondence{\decs'}{\rho}{\ExpMem{\decstilde'}{\mu}}$, $\correspondence{\e'}{\rho}{\ExpMem{\etilde'}{\mu}}$ and
\item {$\rho(\iota_i)=\x_i$} for $i\in 1,\dots,n$.
\end{enumerate}
Applying reduction rule \rn{new} of \refToFigure{conventional} we have that
$\creduce{\ExpMem{\ConstrCall{\C}{\iota_1,\dots,\iota_n}}{\mu}}{\ExpMem{\iota}{\mu'}$ where $\iota\not\in\dom{\mu}$ and $\mu'=\Subst{\mu}{\ConstrCall{\C}{\iota_1,\dots,\iota_n}}{\iota}}$. Then, applying rule  \rn{block} {of \refToFigure{conventional}} we get $\creduce{\ExpMem{\Block{\Dec{\C}{\x}{\iota}\ \decstilde'}{\etilde'}}{\mu'}}{\ExpMem{\Block{\Subst{\decstilde'}{\iota}{\x}}{\Subst{\etilde'}{\iota}{\x}}}{\mu'}}$. Let $\rho'=\Subst{\rho}{\x}{\iota}$, from (iv),
(b), (c) and rule \rn{block} of \refToFigure{correspondence} we get that $\correspondence{\e'_1}{\rho'}{\ExpMem{\Block{\Subst{\decstilde'}{\iota}{\x}}{\Subst{\etilde'}{\iota}{\x}}}{\mu'}}$, and so by \refToLemma{evCtx} we derive $\correspondence{\Ctx{\e'_1}}{\rho'}{\ExpMem{\CtxHat{\Block{\Subst{\decstilde'}{\iota}{\x}}{\Subst{\etilde'}{\iota}{\x}}}}{\mu'}}$ where $\rho\subseteq\rho'$ and $\dom{\mu}\subseteq\dom{\mu'}$.\\
Case 2. In this case from 
$\correspondence{\e_1}{\rho}{\ExpMem{\etilde_1}{\mu}}$, \refToLemma{evCtx} and rule  \rn{block} of \refToFigure{correspondence}
we get 
\begin{enumerate}[(a)]\addtocounter{enumi}{3}
\item $\etilde_1={\Block{\Dec{{\Type{}{\C}}}{\x}{\widehat{\Block{\decs}{\e_b}}}\ \decstilde'}{\etilde'}}$ where 
\item $\correspondence{\dvs\ \dvs'}{\rho}{\mu}$, $\correspondence{\decs'}{\rho}{\ExpMem{\decstilde'}{\mu}}$, $\correspondence{\e'}{\rho}{\ExpMem{\etilde'}{\mu}}$, and $\correspondence{{\Block{\decs}{\e_b}}}{\rho}{\ExpMem{\widehat{\Block{\decs}{\e_b}}}}{\mu}$ 
\end{enumerate}
From (ii), (d), (e) and rule \rn{block} of \refToFigure{correspondence} we derive that  $\correspondence{\e'_1}{\rho}{\ExpMem{\etilde'_1}{\mu}}$ and by \refToLemma{evCtx} we have $\correspondence{\Ctx{\e'_1}}{\rho}{\ExpMem{\CtxHat{\etilde'_1}}{\mu}}$. Since $\ExpMem{\CtxHat{\etilde'_1}}{\mu}$ reduces
in 0 steps to itself we get the result.\\
Rule $\rn{Move-Subterm}$. We consider the value context $\ConstrCall{\C}{{\x_1,\ldots,\x_n},\emptyctx,{\es}}$. The other contexts are similar and easier.
Then
\begin{enumerate}[(i)]
\item $\e_1=\ConstrCall{\C}{{\x_1,\ldots,\x_n},\BlockLab{{\dvs}\ {\dvs'}}{{\x}}{\X},{\es}}$ and
\item $\e'_1={\BlockLab{\dvs}{\ \ConstrCall{\C}{{\x_1,\ldots,\x_n},\BlockLab{{\dvs'}}{{\x}}{\X},{\es}}}{\X\cap\dom{\dvs}}}$.
\end{enumerate}
From $\correspondence{{\e_1}}{\rho}{\ExpMem{\etilde_1}{\mu}}$, (i), rule \rn{new} and rule \rn{block} of \refToFigure{correspondence},
we have that
 \begin{enumerate}[(a)]
 \item $\etilde_1=\ConstrCall{\C}{\iota_1,\dots,\iota_n,\iota,\estilde}$, 
 \item $\correspondence{\dvs\ \dvs'}{\rho}{\mu}$, $\correspondence{{\x_1,\ldots,\x_n}}{\rho}{\ExpMem{\iota_1,\dots,\iota_n}{\mu}}$, $\correspondence{\BlockLab{{\dvs}\ {\dvs'}}{{\x}}{\X}}{\rho}{\ExpMem{\iota}{\mu}}$ and $\correspondence{{\es}}{\rho}{\ExpMem{\estilde}{\mu}}$.
 \end{enumerate}
 From (b) we derive that $\correspondence{\BlockLab{{\dvs'}}{{\x}}{\X}}{\rho}{\ExpMem{\iota}{\mu}}$. Therefore 
 \begin{enumerate}[(a)]\addtocounter{enumi}{2}
  \item $\correspondence{{\BlockLab{\dvs}{\ \ConstrCall{\C}{{\x_1,\ldots,\x_n},\BlockLab{{\dvs'}}{{\x}}{\X},{\es}}}{\X\cap\dom{\dvs}}}}{\rho}{\ExpMem{\ConstrCall{\C}{\iota_1,\dots,\iota_n,\iota,\estilde}}{\mu}}$. 
  \end{enumerate}
From $\correspondence{\ctx}{\rho}{\ExpMem{\ctxHat}{\mu}}$, (i), (a) and \refToLemma{evCtx} we derive that 
$\correspondence{\Ctx{e_1}}{\rho}{\ExpMem{{\CtxHat{\etilde_1}}} {\mu}}$. Since $\ExpMem{{\CtxHat{\etilde_1}}} {\mu}$ 
reduces in 0 steps
to itself we get the result.\\
Consider now \underline{rule \rn{New}}. Then 
\begin{enumerate}[(i)]
\item $\e= \Ctx{\ConstrCall{\C}{\x_1,\dots,\x_n}} $ and
\item $\e'= \Ctx{\BlockLab{\Dec{\C}{\x}{\ConstrCall{\C}{\x_1,\dots,\x_n}}}{\x}{{\{\x\}}}}$.
 \end{enumerate}
 From $\correspondence{\e}{\rho}{\ExpMem{\etilde}{\mu}}$, \refToLemma{evCtx} and rule \rn{new} of \refToFigure{correspondence}
we have that 
\begin{enumerate}[(a)]
\item $\etilde=\CtxHat{\ConstrCall{\C}{\iota_1,\dots,\iota_n}}$ ,
\item  $\correspondence{\ctx}{\rho}{\ExpMem{\ctxHat}{\mu}}$ and 
{$\rho(\iota_i)=\x_i$} for $i\in 1,\dots,n$. 
\end{enumerate}
Applying rule \rn{new} of \refToFigure{conventional} we have that
$\creduce{\ExpMem{\ConstrCall{\C}{\iota_1,\dots,\iota_n}}{\mu}}{\ExpMem{\iota}{\mu'}}$ where 
$\mu'=\Subst{\mu}{\ConstrCall{\C}{\iota_1,\dots,\iota_n}}{\iota}$ and $\iota\not\in\dom{\mu}$. Let $\rho'=\Subst{\rho}{\x}{\iota}$,
from (ii) and rule \rn{block} of \refToFigure{correspondence} we get that 
$ \correspondence{\e'}{\rho'}{\ExpMem{\iota}{\mu'}} $ with $\rho\subseteq\rho'$ and $\dom{\mu}\subseteq\dom{\mu'}$, which proves the result.
\end{proof}

\begin{corollary}\label{preservation}
{If $\FV{\e}=\emptyset$,} $\reducestar{\e}{\val}$, and $\correspondence{\e}{\rho}{\ExpMem{\etilde}{\mu}}$, then $\creducestar{\ExpMem{\etilde}{\mu}}{\ExpMem{\iota}{\mu'}}$ with $\correspondence{\val}{\rho'}{\ExpMem{\iota}{\mu'}}$ for some $\rho', \mu'$ such that $\rho\subseteq\rho'$, and $\dom{\mu}\subseteq\dom{\mu'}$.
\end{corollary}
\begin{proof}
By arithmetic induction on the number of steps.
\begin{description}
\item[Base] If $\reduceN{\e}{\val}{0}$, then $\e=\val$, and $\correspondence{\val}{\rho}{\ExpMem{\etilde}{\mu}}$. Since $\e$ has no free variables, we have that $\val=\BlockLab{\dvs}{\x}{\X}$ for some $\X$, $\dvs$, and $\x$. From \refToLemma{value}.2, $\etilde=\iota$
for some $\iota$ such that $\rho(\iota)=\x$. Therefore $\ExpMem{\iota}{\mu}$ reduces in zero steps to $\ExpMem{\iota}{\mu}$ and the thesis holds.
\item[Inductive step] If $\reduceN{\e}{\val}{n+1}$, then $\reduce{\e}{\e'}$ and $\reduceN{\e'}{\val}{n}$. By inductive hypothesis we have that, if $\correspondence{\e'}{\rho'}{\ExpMem{\etilde'}{\mu'}}$, then $\creducestar{\ExpMem{\etilde'}{\mu'}}{\ExpMem{\iota}{\mu''}}$ and 
$\correspondence{\val}{\rho''}{\ExpMem{\iota}{\mu''}}$, for some $\rho'',\mu''$ such that $\rho'\subseteq\rho''$, $\dom{\mu'}\subseteq\dom{\mu''}$. 
Then, the thesis follows by \refToTheorem{step}. 
\end{description}
\end{proof}


The crucial point  for the preservation theorem is that the substitution semantics, modeling ``moving'' capsules from a location to another in the memory, see rule \rn{affine-elim}, is indeed equivalent to the conventional semantics. Note that this only holds for variables which are used at most once. Consider, for instance, the following reduction sequence in the syntactic calculus, where we omit block annotations for simplicity {and we mention rule \rn{Ctx} only when the context is different from $\emptyctx$}:
\begin{quote}
\begin{math}
\begin{array}{ll}
{\Block{\Dec{\Type{\capsule}{\C}}{\x}{\ConstrCall{\C}{0}}}{\FieldAccess{\x}{\f}}\longrightarrow}&\rn{new} \\
\BlockLab{\Dec{\Type{\capsule}{\C}}{\x}{\BlockLab{\Dec{\C}{\y}{\ConstrCall{\C}{0}}}{\y}{}}}{\FieldAccess{\x}{\f}}{}\longrightarrow&\rn{affine-elim}\\
\FieldAccess{\BlockLab{\Dec{\C}{\y}{\ConstrCall{\C}{0}}}{\y}{}}{\f}\longrightarrow&\rn{move-subterm}\\
\BlockLab{\Dec{\C}{\y}{\ConstrCall{\C}{0}}}{\FieldAccess{\y}{\f}}{}\longrightarrow&\rn{field-access}\\
\BlockLab{\Dec{\C}{\y}{\ConstrCall{\C}{0}}}{0}{}\longrightarrow&\rn{garbage}\\
0
\end{array}
\end{math}
\end{quote}
In the conventional calculus, we get the following corresponding reduction sequence:
\begin{quote}
\begin{math}
\begin{array}{ll}
{\ExpMem{\Block{\Dec{{\C}}{\x}{\ConstrCall{\C}{0}}}{\FieldAccess{\x}{\f}}}{\emptyset}\Longrightarrow}&\rn{new}+\rn{ctx} \\
\ExpMem{\Block{\Dec{\C}{\x}{\iota}}{\FieldAccess{\x}{\f}}}{\iota\mapsto\ConstrCall{\C}{0}}\Longrightarrow&\rn{dec}\\
\ExpMem{\FieldAccess{\iota}{\f}}{\iota\mapsto\ConstrCall{\C}{0}}\Longrightarrow^0\\
\ExpMem{\FieldAccess{\iota}{\f}}{\iota\mapsto\ConstrCall{\C}{0}}\Longrightarrow&\rn{field-access}\\
\ExpMem{0}{\iota\mapsto\ConstrCall{\C}{0}}\Longrightarrow^0\\
\ExpMem{0}{\iota\mapsto\ConstrCall{\C}{0}}
\end{array}
\end{math}
\end{quote}
Consider, instead, a similar example where the affine variable is used twice:

\begin{quote}
\begin{math}
\begin{array}{lll}
\e=&\BlockLab{\Dec{\Type{\capsule}{\C}}{\x}{\BlockLab{\Dec{\C}{\y}{\ConstrCall{\C}{0}}}{\y}{}}}{\Sequence{\FieldAssign{\x}{\f}{3}}{\FieldAccess{\x}{\f}}}{}\longrightarrow&\rn{affine-elim}\\
\e'=&\Sequence{\FieldAssign{\BlockLab{\Dec{\C}{\y}{\ConstrCall{\C}{0}}}{\y}{}}{\f}{3}}{\FieldAccess{\BlockLab{\Dec{\C}{\y}{\ConstrCall{\C}{0}}}{\y}{}}{\f}}\longrightarrow&\rn{move-subterm}+\rn{ctx}\\
&\Sequence{\BlockLab{\Dec{\C}{\y}{\ConstrCall{\C}{0}}}{\FieldAssign{\y}{\f}{3}}{}}{\FieldAccess{\BlockLab{\Dec{\C}{\y}{\ConstrCall{\C}{0}}}{\y}{}}{\f}}\longrightarrow&\rn{field-assign}+\rn{ctx}\\
&\Sequence{\BlockLab{\Dec{\C}{\y}{\ConstrCall{\C}{3}}}{3}{}}{\FieldAccess{\BlockLab{\Dec{\C}{\y}{\ConstrCall{\C}{0}}}{\y}{}}{\f}}\longrightarrow&\rn{garbage}+\rn{ctx}\\
&\Sequence{3}{\FieldAccess{\BlockLab{\Dec{\C}{\y}{\ConstrCall{\C}{0}}}{\y}{}}{\f}}\longrightarrow&\rn{garbage}\\
&\FieldAccess{\BlockLab{\Dec{\C}{\y}{\ConstrCall{\C}{0}}}{\y}{}}{\f}\longrightarrow^\star&\mbox{as above}\\
&0
\end{array}
\end{math}
\end{quote}
{(To understand the second \rn{garbage} reduction, recall that $\Sequence{\e}{\e'}$ is an abbreviation for $\Block{\Dec{\T}{\x}{\e}}{\e'}$ if $\x$ not free in $\e'$.)}

The term $\etilde$ equivalent to $\e$ through $\rho(\iota)=\y$, with the conventional semantics, reduces to  $3$:
\begin{quote}
\begin{math}
\begin{array}{lll}
\etilde=&\ExpMem{\Block{\Dec{\C}{\x}{\iota}}{\Sequence{\FieldAssign{\x}{\f}{3}}{\FieldAccess{\x}{\f}}}}{\iota\mapsto\ConstrCall{\C}{0}}\Longrightarrow&\rn{dec}\\
\etilde'=&\ExpMem{\Sequence{\FieldAssign{\iota}{\f}{3}}{\FieldAccess{\iota}{\f}}}{\iota\mapsto\ConstrCall{\C}{0}}\Longrightarrow&\rn{field-assign}\\
&\ExpMem{\Sequence{3}{\FieldAccess{\iota}{\f}}}{\iota\mapsto\ConstrCall{\C}{3}}\Longrightarrow&\rn{garbage}\\
&\ExpMem{\FieldAccess{\iota}{\f}}{\iota\mapsto\ConstrCall{\C}{3}}\Longrightarrow&\rn{field-access}\\
&\ExpMem{3}{\iota\mapsto\ConstrCall{\C}{3}}
\end{array}
\end{math}
\end{quote}
Hence, in this case the two  reduction sequences are \emph{not} equivalent, and \refToTheorem{step} does not hold.
Indeed, the first reduction step in the syntactic calculus does \emph{not} preserve the matching relation.  Notably, there is no way to find a mapping making terms $\e'$ and $\etilde'$   to match, since such mapping should map $\iota$ in two different variables, one declared in the first block and one in the second (the fact that they are two different variables can be made explicit by $\alpha$-renaming term $\e'$). In more informal words, the substitution semantics \emph{duplicates memory} {(rather than just moving)} if adopted for non-affine variables. 

\section{Conclusion}\label{sect:conclu}
In this paper we presented a calculus for an imperative object-oriented language whose distinguished features are the following.
\begin{itemize}
\item Local variable declarations are used to directly represent the memory. That is, a declared {(non affine)} variable is not replaced by its value, as in standard \texttt{let}, but the association is kept and used when necessary.\item In this way, there are language values (block values) which represent (a portion of) memory, and the fact that such portion of memory is isolated can be \emph{modularly} checked by inspecting  only the value itself, without any need to explore the whole graph structure of the global memory as it would be in the conventional model.
\item To safely handle capsules, the syntactic calculus supports \emph{affine} variables with a special semantics. Runtime checks ensure that their initializing value is a capsule, and their (unique by definition) occurrence is replaced by their capsule value (rule \rn{affine-elim}). 
\item  {\emph{Block annotations} allow a variable declaration to be moved outside of a block only if such variable will \emph{not} be part of its final result.
This additional runtime check prevents to ``break'' capsules before they are moved.}
\end{itemize}

In previous work \cite{CapriccioliSZ15,ServettoZucca15,GianniniSZ16,GianniniSZ17,GianniniSZ18,GianniniSZC19b,GianniniRSZ19a}, we have designed type systems which statically ensure that such runtime checks succeed, hence execution is not stuck. In this paper we {prove}
that the syntactic calculus preserves the standard semantics of imperative calculi relying on 
a global memory. In such calculi reasoning about program properties such as sharing requires the formalization of invariants
on the memory and the proof of their preservation under reduction, whereas in ours this can be done by structural
induction on terms. 

The wider context of the paper is the huge amount of research on mechanisms for controlling sharing and interference. In this paper we focus on the property that the the subgraph reachable from 
a reference $\x$ is an isolated portion of store, that is, all its (non immutable) nodes
can be reached only through this reference. This property has many variants in literature \cite{ClarkeWrigstad03,Almeida97,Hogg91,DietlEtAl07,GordonEtAl12}.
Examples of other relevant properties of referenes studied in literature are the following:
\begin{itemize}
\item $\x$ is \emph{immutable}, that is, the subgraph reachable from 
$\x$ is an \emph{immutable} portion of store. An immutable reference can be safely shared
in a multithreaded environment.
\item $\x$ is \emph{lent}, that is, the subgraph reachable from $\x$ can be manipulated, but not shared, by a client \cite{ServettoZucca15,GianniniSZ16,GianniniSZ18}. This is also called \emph{borrowed} in literature \cite{Boyland01,NadenEtAl12}. 
\item $\x$ is \emph{read-only} if no modification is permitted through $\x$. Note that there is no immutability guarantee.
\end{itemize}

A specular approach to the use of qualifiers to restrict the usage of references is that on \emph{ownership} (see an overview in~\cite{ClarkeEtAl13}), where a formal way is provided to express and prove the ownership invariants. Among the many works in this stream, we mention Rust \cite{Turon17}, which uses ownership {and qualifiers} for memory management.

In future work we plan to add the modelling of immutable
references. We will also investigate (a form of) Hoare logic on top of our model. 

\bigskip

\noindent
{\bf Acknowledgements}\\
We would like to thank the anonymous DCM referees for their helpful comments.

\end{document}